\documentclass[pra,floatfix,amsmath,superscriptaddress,onecolumn,nofootinbib]{revtex4-2}
\usepackage{amssymb}
\usepackage{graphicx}
\usepackage{graphics}
\usepackage{amsmath}
\usepackage{amsthm}
\usepackage{color}

\newcommand{\bea}{\begin{eqnarray}}
\newcommand{\eea}{\end{eqnarray}}

\def\bi{\begin{itemize}}
\def\ei{\end{itemize}}
\def\bc{\begin{center}}
\def\ec{\end{center}}

\def\C{\hbox{$\mit I$\kern-.7em$\mit C$}}
\def\R{\hbox{$\mit I$\kern-.6em$\mit R$}}

\newcommand{\one}{\mbox{$1 \hspace{-1.0mm}  {\bf l}$}}
\def\tr{\mathrm{tr}}

\newtheorem{theorem}{Theorem}

\newtheorem{lemma}[theorem]{Lemma}

\newtheorem{definition}[theorem]{Definition}

\begin{document}

\title{Quantum states with the same entanglement and local unitary equivalence}
\author{Julio I. de Vicente}\email{jdvicent@math.uc3m.es}
\affiliation{Departamento de Matem\'aticas, Universidad Carlos III de
Madrid, E-28911, Legan\'es (Madrid), Spain}
\affiliation{Instituto de Ciencias Matem\'aticas (ICMAT), E-28049 Madrid, Spain}

\begin{abstract}
Local operations assisted by classical communication (LOCC) play a fundamental role in the resource-theoretic formulation of entanglement theory inducing an operationally meaningful ordering in the set of entangled states. The particular class within LOCC of local unitary (LU) transformations is obviously closed under inversion. Thus, LU-equivalent states are interconvertible by LOCC and have therefore the same entanglement. For this reason LU-equivalence has been thoroughly studied in the context of entanglement theory. However, as I note here with a simple example, the converse is not true: there exist entangled states that are interconvertible by LOCC but which are not LU-equivalent. This motivates studying under which conditions LOCC interconvertibility is characterized by LU-equivalence. I provide a sufficient condition for this to hold. When one considers multipartite states with full entanglement dimensionality (as quantified by the Schmidt number vector) LOCC interconvertibility can only hold under LU-equivalence.  
\end{abstract}

\maketitle

\section{Introduction}

Entanglement constitutes a major field of study in quantum science both for its conceptual implications and for underlying different forms of quantum technologies. In the context of quantum information science entanglement is regarded as a resource. Thus, entanglement theory, which is formulated under the umbrella of the so-called quantum resource theories \cite{resource}, aims at characterizing, classifying and quantifying it with respect to its applicability in information-theoretic terms \cite{review1,review2,review3}. A fundamental notion in this theory is the paradigm of local operations and classical communication (LOCC). This makes it possible to study all protocols with which this resource can be manipulated, but, more importantly, it provides an operational ordering in the set of entangled states, which serves as a basic axiomatic ingredient to construct entanglement measures. Since entanglement cannot be created by LOCC, if such a transformation from a state $\rho$ to a state $\sigma$ is possible then $\rho$ cannot be less entangled than $\sigma$, which reflects the fact that $\rho$ cannot be less useful than $\sigma$ for any task to be completed by LOCC manipulation. Consequently, every entanglement measure cannot assign a lower value to $\rho$ than that of $\sigma$.

Given the above, states that are interconvertible by LOCC have exactly the same entanglement properties and entanglement theory puts them in the same equivalence class. Note in particular that all entanglement measures are then invariant on all such states. Thus, one of the most basic questions in this theory is to characterize which states are LOCC-interconvertible since this underlies any further analysis to be performed in this context. One of the simplest LOCC protocols is a local unitary (LU) transformation. Moreover, these transformations are reversible and their inverse is also a LU transformation. Therefore, LU-equivalent states are LOCC-interconvertible and belong to the same entanglement equivalence class. For this reason, this property has been extensively studied in the literature. This includes both establishing standard forms for states up to LU-equivalence and finding means to decide whether two given states are LU-equivalent or not. The well-known Schmidt decomposition provides a simple standard form for all bipartite pure states that are LU-equivalent and several works consider standard forms for pure states beyond the bipartite case \cite{LindenPopescu98,LindenPopescu99,carteret00,Acin00,Acin01}. Regarding the second line of research, the problem of deciding when two pure $n$-qubit are LU-equivalent was solved in \cite{kraus1,kraus2} and for pure $n$-qudit states with $d>2$ in \cite{Liu12}. An alternative line of attack to this problem, which allows one to treat as well the case of mixed states, is given by constructing polynomial invariants, for which it is known that a finite number of them is sufficient \cite{Rains00}. However, a complete list is only known in a few simple cases. This approach is followed by \cite{Grassl97,LindenPopescu99bis} and by \cite{Makhlin02}, which provides a minimal complete list of 18 invariants for the case of 2-qubit mixed states. More recent works tackle this problem in general with the so-called Bargmann invariants \cite{bargmann1,bargmann2,bargmann3} or construct invariants using the generalized Bloch representation \cite{bloch}.

While it is immediate that LU-equivalent states are LOCC-interconvertible, surprisingly little attention seems to have been devoted to the converse. In fact, to my knowledge this question has only been addressed in \cite{Gingrich}, where it is proven (see Theorem 1 therein) that two pure states are LOCC-interconvertible if and only if (iff) they are LU-equivalent. Yet, even though the identification of these two notions in general might appear to be folklore, this still leaves the case of mixed states open. Actually, a minute of thought reveals that LU-equivalence and LOCC-interconvertibility are not the same: all separable (i.e.\ non-entangled) states are interconvertible by LOCC but they are not LU-equivalent. Nevertheless, one might be tempted to conclude that the equivalence remains if one is restricted to consider entangled states. However, a few more minutes of thought show the contrary: a simple example of two entangled states that are LOCC-interconvertible but not LU-equivalent is given below in Sec.\ III after I introduce the necessary mathematical preliminaries in Sec.\ II. Thus, one of the goals of this paper is to point out that in general LOCC-interconvertibility does not imply LU-equivalence and in order to identify states with the same entanglement one needs to go beyond the latter paradigm. This raises the question of how to characterize when two states are LOCC-interconvertible and, in particular, when, as it happens with pure states, it is sufficient to consider the well-studied case of LU-equivalence. The main result of this work, which is presented in Sec.\ IV, is such a sufficient condition. Therein I prove that two $n$-qudit states which are fully entangled in all $d$ levels (as quantified in precise terms by the Schmidt number vector \cite{huberdV13}) are LOCC-interconvertible iff they are LU-equivalent. Section V closes the paper with some concluding remarks and potential future research directions.     

\section{Notation and preliminaries}

This paper considers multipartite quantum systems with $n\geq2$ number of parties, which will be indexed by the elements of $[n]:=\{1,2,\ldots,n\}$, and each party has local dimension denoted by $d\geq2$. Thus, the total Hilbert space associated to the system is $H=\bigotimes_{i=1}^nH_i$ and without loss of generality $H_i=\mathbb{C}^d$ $\forall i$. The set of all density matrices on $H$ will be denoted by $D_{n}(d)$. Given any unit-norm vector $|\psi\rangle\in H$, $\psi\in D_{n}(d)$ stands for its corresponding pure-state density matrix $\psi=|\psi\rangle\langle\psi|$. Multipartite states will sometimes be regarded as bipartite states in some bipartition $M|\overline{M}$ where $\emptyset\neq M\subsetneq[n]$ and $\overline{M}$ is the complement of $M$ in $[n]$. Hence, $D_2(d_1,d_2)$ will denote the set of all density matrices for bipartite quantum systems with local dimension $d_1$ for one party and local dimension $d_2$ for another. 

LOCC protocols are a class of completely positive and trace-preserving (CPTP) maps that correspond to a sequence of unilocal actions by some party that communicates classically the obtained outcome, upon which the other parties implement local unitaries. Thus, one round of an LOCC protocol on $D_n(d)$ where, say, party 1 acts unilocally consists of a choice of a positive operator-valued measure (POVM) $\{K_j\}\subset\mathbb{C}^{d\times d}$ ($\sum_jK_j^\dag K_j=\one$) and of unitary $d\times d$ matrices $\{U_j^{(i)}\}$ ($i=2,3,\ldots,n$) that transforms an input state $\rho\in D_n(d)$ into the state
\begin{equation}
\rho_j=\frac{K_j\otimes U_j^{(2)}\otimes\cdots\otimes U_j^{(n)}\rho K_j^\dag\otimes (U_j^{(2)})^\dag\otimes\cdots\otimes (U_j^{(n)})^\dag}{\tr(K_j\otimes U_j^{(2)}\otimes\cdots\otimes U_j^{(n)}\rho K_j^\dag\otimes (U_j^{(2)})^\dag\otimes\cdots\otimes (U_j^{(n)})^\dag)}
\end{equation}
with probability
\begin{equation}
p_j=\tr(K_j\otimes U_j^{(2)}\otimes\cdots\otimes U_j^{(n)}\rho K_j^\dag\otimes (U_j^{(2)})^\dag\otimes\cdots\otimes (U_j^{(n)})^\dag)=\tr(\rho_1K_j^\dag K_j),
\end{equation}
where $\rho_1$ is the reduced density matrix of $\rho$ corresponding to party 1. Each possible outcome $j$ breaks the LOCC protocol into a branch where the parties can perform further sequences of unilocal actions with classical communication and local-unitary correction starting from $\rho_j$. If an LOCC protocol eventually stops in all branches (i.e.\ finite-round protocols), we say that the procedure transforms the input state $\rho$ into the ensemble corresponding to each output state at every terminating branch with the corresponding overall probability. If each terminating branch produces the same state $\sigma$, we say that $\rho$ can be transformed (deterministically) into $\sigma$ by LOCC, which is denoted by $\rho\to\sigma$. Note that one can also consider infinite-round protocols that never terminate but where the transformation converges to an output ensemble. Further details in this regard can be found in \cite{locc}. 

The simplest example of an LOCC protocol is an LU-transformation, i.e.\ given $\rho\in D_n(d)$, the transformation $\rho\to\sigma=U\rho U^\dag$, where $U=\bigotimes_{i=1}^nU^{(i)}$ and the matrices $\{U^{(i)}\}\subset\mathbb{C}^{d\times d}$ are unitary for all $i$. In this case, there obviously also exists another LU-transformation such that $\sigma\to\rho$ and, therefore, all LU-equivalent states are LOCC-interconvertible. In this paper I will also often consider a particular class of LOCC maps that I term mixed-LU given their analogy to the so-called mixed-unitary channels \cite{watrous}.

\begin{definition}
A CPTP map $\Lambda:D_{n}(d)\to D_{n}(d)$ is called mixed-LU if it can be written in the form
\begin{equation}\label{mixedLU}
\Lambda(\cdot)=\sum_ip_iU_i\cdot U^\dag_i,
\end{equation}
where $U_i=\bigotimes_{j=1}^nV_i^{(j)}$, the matrices $V_i^{(j)}\in\mathbb{C}^{d\times d}$ are unitary for all $i$ and $j$, $p_i>0$ $\forall i$ and $\sum_ip_i=1$.
\end{definition}
Note that every mixed-LU channel is manifestly implementable by LOCC. Furthermore, if every unilocal action in an LOCC protocol corresponds to a probabilistic unitary transformation, i.e.\ $K_j=\sqrt{q_j}U_j$ in the notation of the discussion around Eqs.\ (1)-(2) where the $\{q_j\}$ gives rise to a probability distribution and $U_j$ is a unitary matrix for all $j$, then the corresponding CPTP map is mixed-LU. Thus, if an LOCC map is not mixed-LU there has to be at least one party at some point in the protocol whose unilocal action is not a probabilistic unitary transformation.

It is well-known (see e.g.\ \cite{nielsenchuang}) that every pure state $\psi\in D_{2}(d_1,d_2)$ can be put by LU transformations in its Schmidt form
\begin{equation}
|\psi\rangle=\sum_{i=1}^{\min(d_1,d_2)}\sqrt{\lambda_i^\psi}|ii\rangle,
\end{equation}
where the so-called Schmidt coefficients, which are non-negative and arranged in non-increasing order, give rise to the vector $\vec{\lambda}^\psi$. The celebrated Nielsen's theorem characterizes LOCC convertibility for bipartite pure states \cite{nielsen}. Namely, given any two pure states $\psi,\phi\in D_{2}(d_1,d_2)$, $\psi\to\phi$ iff $\vec{\lambda}^\phi\succ\vec{\lambda}^\psi$, where $\succ$ denotes majorization \cite{MO}, i.e.\ 
\begin{equation}
\sum_{i=1}^k\lambda_i^\phi\geq\sum_{i=1}^k\lambda_i^\psi\quad\forall k.
\end{equation}
A function $f$ over probability distributions for which it holds that $f(\vec{q})\leq f(\vec{p})$ if $\vec{q}\succ\vec{p}$ is called Schur-concave. Consequently, all Schur-concave functions define good entanglement measures on bipartite pure states. The most paradigmatic case is the entanglement entropy, given by
\begin{equation}\label{entropy}
E(\psi)=H(\vec{\lambda}^\psi),
\end{equation}
where $H$ denotes the Shannon entropy, which is actually known to be strictly Schur-concave and strictly concave (see e.g.\ \cite{MO}). This is generalized to mixed states $\rho\in D_2(d_1,d_2)$ by the entanglement of formation \cite{EoF}
\begin{equation}
E(\rho)=\inf_{\{p_i,\psi_i\}}\sum_i p_i H(\vec{\lambda}^{\psi_i}),
\end{equation}
where the optimization runs over all ensembles of pure states realizing $\rho$, i.e.\ over all $\{p_i,\psi_i\}$ such that $\rho=\sum_ip_i|\psi_i\rangle\langle\psi_i|$ with $p_i>0$ $\forall i$ and $\sum_ip_i=1$. The entanglement of formation is a bona fide entanglement measure on $D_{2}(d_1,d_2)$, i.e.\ if $\rho\to\sigma$, then $E(\rho)\geq E(\sigma)$. When $\rho\in D_n(d)$ I denote by $E_i(\rho)$ the entanglement of formation in the bipartition $i|\overline{i}$ for any party $i\in[n]$. Note that for every $i\in[n]$, $E_i(\cdot)$ is non-increasing by $n$-partite LOCC manipulation.

In addition to the entanglement entropy and the entanglement of formation, in this paper the following entanglement quantifiers will also be considered. Given any pure state $\psi\in D_{2}(d_1,d_2)$, its Schmidt rank $SR(\psi)$ stands for the number of non-zero entries in $\vec{\lambda}^\psi$. This apprehends how many quantum levels are effectively entangled. This notion is extended to mixed bipartite states by considering the Schmidt number \cite{schmidtnumber}, which is defined for $\rho\in D_2(d_1,d_2)$ as
\begin{equation}
SN(\rho)=\min_{\{p_i,\psi_i\}}\max SR(\psi_i),
\end{equation}
where the minimization runs again all ensembles of pure states realizing $\rho$. In the case of multipartite states this information is captured by the Schmidt number vector \cite{huberdV13}. Given now any pure state $\psi\in D_{n}(d)$, similarly as above, for every $i\in[n]$ we can define $SR_i(\psi)$ to be the Schmidt rank of the vector of Schmidt coefficients $\vec{\lambda}^\psi_i$ corresponding to the bipartition $i|\overline{i}$. Arranging these numbers in non-increasing order one constructs the Schmidt rank vector $SRV(\psi)\in\mathbb{N}^n$ (note therefore that the $i$th entry of this vector, $SRV_i(\psi)$, is then fixed by the ordering constraint and it needs not be $SR_i(\psi)$). Thus, for arbitrary $\rho\in D_n(d)$ the single-party Schmidt number vector is defined by the vector $SNV(\rho)\in\mathbb{N}^n$ with entries given by
\begin{equation}
SNV_j(\rho)=\min_{\{p_i,\psi_i\}}\max SRV_j(\psi_i).
\end{equation}
Note in particular that if $\rho\in D_n(d)$ has maximal single-party Schmidt number vector, i.e.\ $SNV(\rho)=(d,\ldots,d)$, then for every pure-state ensemble decomposition such that $\rho=\sum_ip_i\psi_i$ there always exists at least one choice for $i$ such that $SRV(\psi_i)=(d,\ldots,d)$.

Finally, in Lemma \ref{lemma1} below we use that for $\rho\in D_n(d)$, the von Neumann entropy $S(\rho)=-\tr(\rho\log\rho)$ is not only known to be concave but also strictly concave (see e.g.\ \cite{carlen}).

\section{LOCC-interconvertibility does not imply LU-equivalence}

As explained in the introduction, in this section I provide a simple example of two states in $D_2(4)$ that are LOCC-interconvertible but not LU-equivalent. Define
\begin{equation}
|\phi^+_{12}\rangle=\frac{1}{\sqrt{2}}(|11\rangle+|22\rangle),\quad|\phi^+_{34}\rangle=\frac{1}{\sqrt{2}}(|33\rangle+|44\rangle)
\end{equation}
and $\rho\in D_2(4)$ by
\begin{equation}
\rho=\frac{1}{2}\phi^+_{12}+\frac{1}{2}\phi^+_{34}.
\end{equation}
Now, it is easy to see that $\phi^+_{12}$ and $\rho$ are LOCC-interconvertible. However, one is pure and the other is not. Therefore, they are not LU-equivalent. To see that $\phi^+_{12}\to\rho$, it suffices to note that $\Lambda(\phi^+_{12})=\rho$ for the mixed-LU map
\begin{equation}
\Lambda(X)=\frac{1}{2}X+\frac{1}{2}U\otimes U X U^\dag\otimes U^\dag,
\end{equation}
where $U$ is the $4\times4$ unitary matrix such that permutes $|1\rangle$ with $|3\rangle$ and $|2\rangle$ with $|4\rangle$. To see that $\rho\to\phi^+_{12}$, consider a one-round protocol in which party 1 acts unilocally implementing the two-outcome POVM
\begin{equation}
K_1=|1\rangle\langle1|+|2\rangle\langle2|,\quad K_2=|3\rangle\langle3|+|4\rangle\langle4|.
\end{equation}
Then, on outcome 1 this branch of the protocol terminates. If outcome 2 is obtained, the parties implement the LU-transformation given by $U\otimes U$, where $U$ is the unitary matrix above.

\section{A sufficient condition under which LOCC-interconvertibility implies LU-equivalence}

In this section I present the main result of this paper. If two multipartite states have full entanglement dimensionality, then indeed LOCC-interconvertibility and LU-equivalence are equivalent notions. This is precisely stated in the following theorem.

\begin{theorem}\label{th}
Let $\rho\in D_{n}(d)$ and $\sigma\in D_{n}(d)$ satisfy that $SNV(\rho)=SNV(\sigma)=(d,d,\ldots,d)$. If $\rho$ and $\sigma$ are interconvertible by LOCC, then $\rho$ and $\sigma$ are LU-equivalent.
\end{theorem}

The proof of the theorem is presented in the following, for which we derive first some lemmas.

\begin{lemma}\label{lemma1}
Let $\rho\in D_{n}(d)$ and $\sigma\in D_{n}(d)$. If there exist mixed-LU maps $\Lambda$ and $\Lambda'$ such that $\Lambda(\sigma)=\rho$ and $\Lambda'(\rho)=\sigma$, then $\rho$ and $\sigma$ are LU-equivalent.
\end{lemma}
\begin{proof}
Due to the concavity of the Von Neumann entropy, the assumption implies that $S(\sigma)\leq S(\rho)\leq S(\sigma)$. Thus, the inequalities here must hold with equality. But, if $\Lambda$ is a mixed-LU channel ($\Lambda(\cdot)=\sum_ip_iU_i\cdot U^\dag_i$) and $\Lambda(\sigma)=\rho$, due to the strict concavity of Von Neumann entropy, then it must hold that $S(\rho)>S(\sigma)$ unless $U_i\sigma U^\dag_i=\rho$ $\forall i$. This proves the claim. 
\end{proof}

\begin{lemma}\label{lemma2}
Let $\psi\in D_{n}(d)$ be a pure state with $SNV(\psi)=(d,d,\ldots,d)$ and $\rho\in D_{n}(d)$. If $\Lambda:D_{n}(d)\to D_{n}(d)$ is an LOCC transformation that is not mixed-LU and $\Lambda(\psi)=\rho$, then it must hold that $E_j(\psi)>E_j(\rho)$ for some $j\in[n]$ (which must happen for every $j$ corresponding to a party that does not implement a probabilistic unitary transformation at some step of the protocol)
\end{lemma}
\begin{proof}
Since $\psi$ is pure, then $\psi\to\rho$ by LOCC iff $\psi\to\{p_i,\psi_i\}$ by LOCC, where $\{p_i,\psi_i\}$ is some ensemble of pure states realizing $\rho$. At the same time, the existence of an LOCC transformation from $\psi$ to $\rho$ implies that there exists a bipartite LOCC transformation that implements the same conversion when we view the states as bipartite states in the bipartition $j|\overline{j}$ for any $j\in[n]$. For any such fixed choice of $j$ I refer to such protocols as bipartite $j$-LOCC. Reference \cite{JPensemble} (see also \cite{nielsenvidal}) has shown that $\psi\to\{p_i,\psi_i\}$ by bipartite $j$-LOCC iff $\sum_ip_i\vec\lambda_j^{\psi_i}$ majorizes $\vec\lambda_j^\psi$, where the subindex $j$ indicates that Schmidt coefficients correspond to the bipartition $j|\overline{j}$. Thus, $\psi\to\rho$ by bipartite $j$-LOCC implies that for some ensemble of pure states $\{p_i,\psi_i\}$ realizing $\rho$ it holds that
\begin{align}
E_j(\psi)=H(\vec\lambda_j^\psi)&\geq H(\sum_ip_i\vec\lambda_j^{\psi_i})\nonumber\\
&\geq \sum_ip_iH(\vec\lambda_j^{\psi_i})\nonumber\\
&\geq E_j(\rho),\label{eq1lemma2}
\end{align}
where in the first inequality we have used Schur-concavity of $H$ and in the second inequality concavity. But, since we actually have strict concavity, the inequality in the second line of (\ref{eq1lemma2}) is strict if $\vec\lambda_j^{\psi_i}\neq \vec\lambda_j^{\psi_k}$ for some $i\neq k$. On the other hand, strict Schur-concavity imposes that the first inequality in (\ref{eq1lemma2}) is strict if $\vec\lambda_j^\psi\neq\sum_ip_i\vec\lambda_j^{\psi_i}$. Therefore, unless $\vec\lambda_j^{\psi_i}=\vec\lambda_j^\psi$ $\forall i$, we obtain that $E_{j}(\psi)>E_{j}(\rho)$. Therefore, to finish the proof we have to show that for any $n$-partite LOCC protocol such that $\psi\to\{p_i,\psi_i\}$ for which the condition $\vec\lambda_j^{\psi_i}=\vec\lambda_j^\psi$ holds $\forall j,i$ (which, as we have just seen, is necessary to fulfill that $E_j(\psi)=E_j(\rho)$ $\forall j$) must be mixed-LU (note that the fact that a mixed-LU map satisfies the condition is immediate, but we have to prove it is the only LOCC map satisfying it).

Now, as pointed out in Sec.\ II, $n$-partite LOCC protocols consist of a sequence of unilocal actions with local-unitary correction in which each possible outcome splits the protocol in different branches that dictate the subsequent unilocal operations. Without loss of generality suppose that party 1 acts first unilocally transforming $|\psi\rangle$ into $|\phi_k\rangle$ with probability $q_k$ if outcome $k$ is obtained. If $\vec\lambda_1^{\phi_k}\neq\vec\lambda_1^\psi$ for some $k$, repeating the argument of Eq.\ (\ref{eq1lemma2}) we then obtain that $E_1(\psi)>\sum_kq_kE_{1}(\phi_k)$. By completing the protocol in each branch, each state $\phi_k$ is finally converted by further LOCC into a state $\sigma_k$ such that $\rho=\sum_kq_k\sigma_k$. But, then
\begin{equation}\label{reduction}
E_1(\psi)>\sum_kq_kE_{1}(\phi_k)\geq\sum_kq_kE_1(\sigma_k)\geq E_1(\sum_kq_k\sigma_k)=E_1(\rho).
\end{equation}
Therefore, in order to have $E_1(\psi)=E_1(\rho)$ at the end of the protocol, we must necessarily require at this first step that $\vec\lambda_1^{\phi_k}=\vec\lambda_1^\psi$ $\forall k$, i.e.\ $|\phi_k\rangle=U_k\otimes V_k|\psi\rangle$ with the $d\times d$ matrices $\{U_k\}$ and the $d^{n-1}\times d^{n-1}$ matrices $\{V_k\}$ all unitary. 

Thus, consider an arbitrary unilocal action by the first party with POVM $\{K_k\}$
and impose the property established above that
\begin{equation}\label{translemma2}
K_k\otimes\one|\psi\rangle=\sqrt{q_k}U_k\otimes V_k|\psi\rangle\quad\forall k.
\end{equation}
I will now show using the assumption that $SNV(\psi)=(d,d,\ldots,d)$ that this implies that $K_k\in\mathbb{C}^{d\times d}$ must be proportional to a unitary matrix for each $k$. For this, I use that a our state can be written as 
\begin{equation}
|\psi\rangle=A\otimes\one|\phi_+\rangle, 
\end{equation}
where $|\phi_+\rangle$ is the bipartite maximally entangled state with local dimension $d^{n-1}$ and $A\in\mathbb{C}^{d\times d^{n-1}}$. Crucially, the premise on the Schmidt number vector of $\psi$ implies that $A$ is full-rank. Therefore, its pseudoinverse $A^{-1}\in\mathbb{C}^{d^{n-1}\times d}$ satisfies that $AA^{-1}=\one$, while $A^{-1}A$ is a projection on a $d$-dimensional subspace (the support of $A^\dagger A$). Note that this also implies that $q_k\neq0$ for all $k$. In addition to the above, I will also use that $S_1\otimes S_2|\phi_+\rangle=|\phi_+\rangle$ iff $S_2=S_1^{-T}$ (see e.g.\ \cite{gourwallach}). With all these observations, Eq.\ (\ref{translemma2}) is equivalent to
\begin{equation}
U_k^\dag K_kAV_k^\ast\otimes\one|\phi_+\rangle=\sqrt{q_k}A\otimes\one|\phi_+\rangle,
\end{equation}
and, then, to
\begin{equation}
U_k^\dag K_kAV_k^\ast=\sqrt{q_k}A\Rightarrow K_k=\sqrt{q_k}U_kAV_k^TA^{-1},
\end{equation}
which leads to
\begin{equation}\label{kdagak}
K_k^\dag K_k=q_kA^{-\dagger}V_k^\ast A^\dag A V_k^TA^{-1}.
\end{equation}
Hence, $\sum_kK_k^\dag K_k=\one$ implies that
\begin{equation}\label{final}
(A^{-1}A)(\sum_kq_kV_k^\ast A^\dag A V_k^T)(A^{-1}A)=A^\dag A.
\end{equation}
But now, $\sum_kq_kV_k^\ast A^\dag A V_k^T$ can be thought of as the output of a mixed-unitary channel on the positive semidefinite matrix $A^\dag A$. Therefore, the ordered eigenvalues of the former are strictly majorized by those of the latter unless $V_k^\ast A^\dag A V_k^T=V_{k'}^\ast A^\dag A V_{k'}^T$ $\forall k\neq k'$ \cite{watrous}. Additionally, left and right-multiplying a positive semidefinite matrix with a projection cannot increase any of the ordered eigenvalues (which can be easily seen as a consequence of the Courant-Fischer min-max theorem; see e.g.\ \cite{bhatia}). Hence, the only way in which Eq.\ (\ref{final}) can hold is if $V_k^\ast A^\dag A V_k^T=V_{k'}^\ast A^\dag A V_{k'}^T$ $\forall k\neq k'$, which must have additionally the same support as $A^\dag A$. Therefore, $V_k^\ast A^\dag A V_k^T=A^\dag A$ $\forall k$ and, consequently, by Eq.\ (\ref{kdagak}) it follows that $K^\dag_k K_k=q_k\one$, as I wanted to show.

Summing up, I have proven that unless the first unilocal action corresponds to party 1 implementing a local unitary $U_k$ with probability $q_k$, Eq.\ (\ref{reduction}) applies and it will hold at the end of the protocol that $E_1(\psi)>E_1(\rho)$. Thus, in order for equality to hold, this entails that the subsequent states at every branch that arise from this first step of the protocol must all be LU-equivalent to $\psi$. This means that we can now reiterate the argument for each of these states concluding that the next unilocal action in each branch independently of which party $j$ implements it must be of the same form so that $E_j$ does not decrease at the end of the protocol. Using this argument iteratively we conclude that most general LOCC protocol that transforms $\psi$ to $\rho$ fulfilling that $E_j(\psi)=E_j(\rho)$ $\forall j$ must correspond to the parties implementing different LU transformations with a certain probability and, hence, a mixed-LU channel\footnote{Note that this applies as well even if the LOCC protocol has infinitely many rounds as the unitary group is compact.}. This concludes the proof.
\end{proof}

\begin{proof}[Proof of Theorem~\ref{th}]
If $\rho$ and $\sigma$ are interconvertible by mixed-LU maps, then we are done by Lemma \ref{lemma1}. So let us assume that one of the LOCC maps cannot be mixed-LU. Without loss of generality, say that any LOCC map $\Lambda$ such that $\Lambda(\sigma)=\rho$ is not mixed-LU. Denote by $\{p_i,\psi_i\}$ an arbitrary pure-state ensemble realizing $\sigma$. Then, by linearity, it must hold that $\sum_ip_i\Lambda(\psi_i)=\rho$, which implies that
\begin{equation}
\Lambda(\psi_i)=\rho_i,\quad \sum_ip_i\rho_i=\rho.
\end{equation}
Since $\Lambda$ is not mixed-LU and our assumption that $SNV(\sigma)=(d,d,\ldots,d)$ implies in this case that for every ensemble decomposition of $\sigma$ it holds that $SNV(\psi_i)=(d,d,\ldots,d)$ for some $i$, we can apply Lemma \ref{lemma2} to conclude that there exists a fixed $j$ independently of the ensemble decomposition (corresponding to a party that does not implement a probabilistic unitary transformation at some step of the protocol $\Lambda$) such that $E_j(\psi_i)> E_j(\rho_i)$ for any such $i$. Therefore, for every pure-state ensemble $\{p_i,\psi_i\}$ that realizes $\rho$ we have a fixed $j\in[n]$ such that
\begin{equation}
\sum_ip_iE_{j}(\psi_i)>\sum_ip_iE_{j}(\rho_i)\geq E_{j}(\sum_ip_i\rho_i)=E_{j}(\rho),
\end{equation}
where I have used the convexity of $E_{j}$. Since this holds for every ensemble decomposition of $\sigma$, we arrive then at the contradiction $E_{j}(\sigma)>E_{j}(\rho)$, as this implies that $\rho$ cannot be converted by LOCC into $\sigma$.
\end{proof}

\section{Conclusion and further work}

LU equivalence has been actively investigated in the context of entanglement theory given that it implies LOCC interconvertibility and it therefore identifies states in the same entanglement equivalence class. Although one might have been tempted to think that these two notions are interchangeable, I have pointed out here with a simple example that this is not the case. There exist pairs of states which are not LU-equivalent even though they are LOCC-interconvertible. Nevertheless, this and other examples that might come to mind are based on a very particular structure such as dimension-deficient states. Intuition suggests that excluding highly particular states the above identification could be restored. This is precisely the main contribution of this work. It provides a complete and rigorous proof that for states with full entanglement dimensionality LOCC interconvertibility holds iff they are LU-equivalent. 

While this work proves that the aforementioned condition of full entanglement dimensionality is sufficient for this two notions to be the same, it seems quite likely that it is not necessary. For the future it would be interesting to investigate which precise conditions characterize the particular states for which LOCC interconvertibility does not boil down to LU equivalence. Also, in this paper I considered multipartite states with the same local dimension, $D_n(d)$, while one could study quantum systems with different local dimensions, i.e.\ $D_n(d_1,\ldots,d_n)$. Note that the proof technique put forward here carries on straightforwardly to this case as long as states for which every pure-state ensemble decomposition contains a pure state with Schmidt number vector equal to $(d_1,\ldots,d_n)$ are considered. This in particular covers the case of pure states already proven in \cite{Gingrich}. Note, however, that maximal Schmidt number vector is no longer a sufficient condition for the identification of LOCC interconvertibility and LU equivalence to hold in the case of uneven local dimensions. For instance, in $D_2(4,2)$ the maximal possible Schmidt number is two, which is the case of the states $\phi^+_{12},\rho\in D_2(4,2)$, where
\begin{equation}
|\phi^+_{12}\rangle=\frac{1}{\sqrt{2}}(|11\rangle+|22\rangle),\quad|\phi^+_{34}\rangle=\frac{1}{\sqrt{2}}(|31\rangle+|42\rangle),\quad
\rho=\frac{1}{2}\phi^+_{12}+\frac{1}{2}\phi^+_{34}.
\end{equation}   
A similar argument as the one in Sec.\ III readily shows that $\phi^+_{12}$ and $\rho$ are LOCC-interconvertible but not LU-equivalent.

\begin{acknowledgments}
I wish to thank Barbara Kraus for long-lasting discussions over the years on this and related topics. I acknowledge financial support from the Spanish Ministerio de Ciencia, Innovaci\'on y Universidades (grant PID2023-146758NB-I00, grant PID2024-160539NB-I00, and ``Severo Ochoa Programme for Centres of Excellence" grant CEX2023-001347-S all funded by MCIN/AEI/10.13039/501100011033) and from Comunidad de Madrid (grant QUITEMAD-CM TEC-2024/COM-84).
\end{acknowledgments}

\end{document}